\documentclass{article}
\usepackage{a4wide,amsmath,amssymb,amsfonts}
\newenvironment{proof}{\textbf{Proof.~}}{\hfill$\Box$}

\usepackage[dvips,usenames,dvipsnames]{xcolor}

\usepackage{graphicx,psfrag,epsfig}
\usepackage{url,xspace}

\newcommand{\ie}{i.e.~}
\newcommand{\eg}{e.g.,~}

\newcommand{\MATH}[1]{\ensuremath{#1}\xspace}
\newcommand{\MATHIT}[1]{\MATH{\mathit{#1}}}
\newcommand{\MATHBB}[1]{\MATH{\mathbb{#1}}}

\newcommand{\AXIOM}[1]{\MATH{\mathsf{#1}}}
\newcommand{\AxSameFuture}{\AXIOM{AxSameFuture}}
\newcommand{\AxCMV}{\AXIOM{AxCoMoving}}
\newcommand{\AxSTL}{\AXIOM{AxSTLMotion}}
\newcommand{\SPR}{\AXIOM{SPR}}

\newcommand{\THEORY}[1]{\AXIOM{#1}}

\newcommand{\HORIZONTAL}{\MATHIT{horizontal}}
\newcommand{\VERTICAL}{\MATHIT{vertical}}
\newcommand{\LEFT}{\MATHIT{left}}
\newcommand{\RIGHT}{\MATHIT{right}}
\newcommand{\BOTTOM}{\MATHIT{bottom}}
\newcommand{\TOP}{\MATHIT{top}}

\newcommand{\REAL}{\MATHBB{R}}
\newcommand{\SPACETIME}{\MATH{\REAL^1_1}}
\newcommand{\nSPACETIME}{\MATH{\REAL^n_1}}
\newcommand{\tSPACETIME}{\MATH{\REAL^3_1}}

\newcommand{\QED}{\MATH{\hfill\Box}}

\newtheorem{theorem}{Theorem}

\begin{document}

\title{Faster than light motion does not imply time travel}

\author{H. Andr\'eka$^1$, J. X. Madar\'asz$^1$, I. N\'emeti$^1$, M. Stannett$^2$ and G. Sz\'ekely$^1$\\[12pt]
   $^1$Alfr\'ed R\'enyi Institute of Mathematics, Hungarian Academy of Sciences,\\
   Re\'altanoda utca 13-15, H-1053, Budapest, Hungary.\\
   \small{\{\texttt{andreka.hajnal}, \texttt{madarasz.judit}, \texttt{nemeti.istvan}, \texttt{szekely.gergely}\}$@$\texttt{renyi.mta.hu}}\\[12pt]
   $^2$Department of Computer Science, University of Sheffield,\\
   Regent Court, 211 Portobello, Sheffield S1~4DP, UK.\\
   \small{\texttt{m.stannett$@$sheffield.ac.uk}}
}

\date{}



\maketitle

\begin{abstract}
Seeing the many examples in the literature of causality violations based on faster-than-light
(FTL) signals one naturally thinks that FTL motion leads inevitably to the
possibility of time travel. We show that this logical inference is invalid by demonstrating a
model, based on (3+1)-dimensional Minkowski spacetime, in which FTL motion is permitted (in every direction without any limitation on speed) yet which
does not admit time travel. Moreover, the Principle of Relativity is true in this model in the sense
that all observers are equivalent. In short, FTL motion does not imply time travel after all.
\end{abstract}



\section{Introduction}
\label{sec:introduction}

The idea that faster-than-light (FTL) motion leads to causality violations goes back at least as far as Einstein \cite{Ein07} and Tolman \cite[pp.~54--55]{Tol17}, although Recami \cite{Recami87} has pointed out that most of the related paradoxes do not involve a valid sequence of causal influences because some of the observers involved disagree as to whether a given event represents the emission or receipt of a signal. Nonetheless, Newton has described a convincing scenario in which two observers $o_1$ and $o_2$ can send FTL signals $s_1, s_2$ to one another, they agree on the temporal order of the times of sending and receiving these two signals, yet each of them receives a reply to their signal before it is sent (see Figure \ref{ccc}). Encoding messages within these signals (assuming this arrangement to be valid) would therefore enable the observers to create a logically paradoxical formation. These FTL-based paradoxes can in principle be resolved by analogy with similar paradoxes proposed in relation to closed timelike curves (CTCs), \eg by appealing to Novikov's self-consistency principle \cite{Nov90,Nov92}, but nonetheless it is natural to ask whether sending information back to the past (which we will simply call `time travel')  must be possible if particles can move faster than light relative to one another.
\begin{figure}[h!]
\centering
\includegraphics*{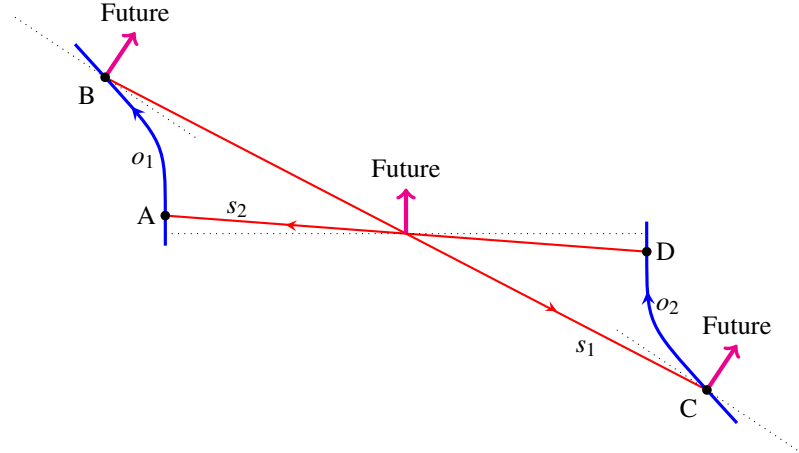}
\caption{Newton's example  of a ``closed causal cycle'' \cite{New70} generated using FTL signals. At event $B$ observer $o_1$ sends out FTL signal $s_1$. This is received by observer $o_2$ at event $C$, who later sends out FTL signal $s_2$. This is received by $o_1$ at event $A$, which is earlier than $B$ from $o_1$'s point of view. The dashed lines illustrate the simultaneities of $o_1$ and $o_2$ where they send and receive these signals.}
\label{ccc}
\end{figure}

The `time travel' capabilities suggested by scenarios like that in Newton's example give force to several results at the interface between physics and computation. For example, Deutsch's seminal quantum computational analysis of CTCs \cite{Deu91} is based on their use as a mechanism for time travel, and his analysis applies (with some restrictions) to any system in which negative time-delay components can be implemented for use in otherwise standard circuits. Likewise, several remarkable complexity theoretic results based on various formulations of CTC computation become relevant \cite{Bru03,BW12,Llo11}, as do related classical formulations of phenomena like wormholes \cite{Nov92} that are of considerable interest in classical computability theory \cite{PC10,Sta06,Sta12}. Akl, for example, has re-opened the debate on universal computation by showing that even when equipped with the capacity for time travel there are basic problems that cannot be solved using a pre-programmed machine \cite{Akl10}. At the other extreme, various authors have shown how the existence of wormholes, CTCs, Malament-Hogarth spacetimes and the like can boost computational power to the extent that formally undecidable problems become solvable \cite{ANS12,EN02,Hog04}. Conversely, computational considerations can be used to indicate certain information theoretic features of CTCs, including both their likely ability to act as information-storage devices and their limitations in this regard \cite{Sta13}.

These results cannot be reduced to the existence of FTL particles, however, because as we demonstrate here, a logical analysis of ($n$+1)-dimensional Minkowski spacetime shows -- subject to a small number of reasonable assumptions identified below -- that, even if we allow the existence of interacting observers/particles traveling FTL relative to one another, this is not in itself enough to entail the possibility of time travel.
We first show in (1+1)-dimensional spacetime (Theorem~\ref{thm:main}) that time travel is possible using FTL observers/particles (in the sense of sending information back to their own pasts) only if it is also possible without using them. The proof of Theorem~\ref{thm:main} gives us an extension of the standard model of special relativity with FTL observers on all spacelike lines, in which Einstein's Principle of Relativity holds but time travel is not possible. We then generalize this model to include each ($n$+1)-dimensional spacetime ($n = 1, 2, 3, \dots$), thereby showing (Theorem~\ref{thm:mainplus}) that FTL motion does not \textit{in itself} introduce the possibility of time travel after all. This does not invalidate the results cited above; it merely rules out the existence of FTL motion as a \emph{logically sufficient} mechanism for explaining or implementing them in practice.

We have tried to keep the following discussion rather informal, but it should be emphasized that our results can also be derived in the context of the first-order axiomatic theory \THEORY{AccRel} of special relativity with accelerating observers; for formal discussions of \THEORY{AccRel} and the related theory \THEORY{SpecRel} of inertial observers, see, \eg \cite{AMN02,MNS06,AMNSz12,SN12}.

\section{FTL motion without time travel}
\label{sec:ftl-travel-without-time-travel}

We begin by working in (1+1)-dimensional Minkowski spacetime. We enrich this spacetime with a large number of FTL observers (each non-light-like line will be a worldline) in such a way that time travel is still not possible. Then we extend this model to ($n$+1)-dimensional Minkowski spacetime, for all $n \geq 1$.

\subsection{The (1+1)-dimensional case}
\label{sec:1+1} 

We focus initially on (1+1)-dimensional Minkowski spacetime, \SPACETIME, where for the sake of argument we will assume \emph{a priori} that it is possible for observers with reference frames attached to them to move FTL relative to one another. We can make this assumption since it is both well-known and indeed easy to prove that FTL motion of observers is compatible with special relativity in \SPACETIME, because there is no real geometric (and hence no physically relevant) distinction between the `inside' and `outside' of a lightcone in this setting. For a formal construction, see, \eg \cite[\S~2.7]{AMN02}.

As usual, we will think of \SPACETIME in terms of a spacetime diagram (coordinatized initially by the reference frame of some inertial observer, $o$, whose identity need not concern us), and note that any lightcone divides \SPACETIME into 4 quadrants, which we will refer to informally in the obvious way using the terms \LEFT, \RIGHT, \TOP and \BOTTOM. Since we want to compare how different observers experience the direction of time's arrow, we identify observers in terms of their worldlines. Technically, we regard each worldline as a smooth path in \SPACETIME (\ie a smooth function $w \colon \REAL \to \SPACETIME$ parametrized by arc length) swept out (as observed by $o$) by the associated observer as it moves continuously into (what it considers to be) its unfolding future.

We now impose the usual inertial approximation condition, where an inertial observer in \SPACETIME is one whose worldline is a straight line that is not tangential to a lightcone -- inertial observers neither accelerate, nor travel at light speed (though they might travel FTL). The following is an informal version of the corresponding formal axiom in \cite{MNS06}:
\begin{itemize}
\item 
  \AxCMV \\
  At each event along an observer's worldline, there exists precisely one co-moving inertial observer.
\end{itemize}
Consequently, not only does the slope of $w$'s tangent vary continuously as we move along $w$ (because $w$ is assumed to be smooth), but $w$ is never tangential to a lightcone. This also implies that no worldline can accelerate from below the speed of light to above the speed of light.

Suppose, then, that $e$ is an event on some worldline $w$. Since $w$'s tangent at $e$ is a straight line which is not tangential to the lightcone at $e$, it must either lie in the region \VERTICAL = $\TOP \cup \{e\} \cup \BOTTOM$ or else in the region \HORIZONTAL = $\LEFT \cup \{e\} \cup \RIGHT$. Supposing for the sake of argument that the tangent at $e$ lies in \VERTICAL, it will also lie in \VERTICAL at every other point along $w$, since otherwise -- by application of the Intermediate Value Theorem to the tangent's slope -- there would necessarily exist some event on $w$ where the motion is light-like. Having established that the tangent lies in \VERTICAL, the same reasoning then allows us to identify uniquely whether $w$ (viewed as an evolving trajectory) unfolds into \TOP or \BOTTOM, and once again the smoothness of motion ensures that this determination will be the same at all events along $w$.

It is therefore meaningful to identify one of the quadrants \LEFT, \RIGHT, \TOP or \BOTTOM as $w$'s \emph{future quadrant}, viz. the quadrant into which the $w$'s unfolding trajectory takes it at every event along its worldline. When we say that `time flows in the same direction'  for two observers, we mean simply that they have the same future quadrant. We call two observers \emph{slower than light} (STL) relative to each other if they are both in \VERTICAL or in \HORIZONTAL as defined above. They are called FTL otherwise. We say that time flows in the opposite direction for two observers if these two observers are STL relative to each other and time does not flow in the same direction for them. We note that these notions are all Lorentz-invariant (or, in other words, observer-independent) in the sense that if two observers are, \eg STL relative to each other, then they remain so after any Lorentz transformation.

We also make the following assumption, which follows from a natural generalization of axiom \AXIOM{Ax5^+} from \cite[p.~297]{AMN02}: 
\begin{itemize}
\item \AxSTL\\
For any observer $o_1$ and worldline $w$ which is STL with respect to $o_1$, there is an observer $o_2$ whose time flows in the same direction as $o_1$'s, and whose worldine and $w$ have the same range.
\end{itemize}
This axiom allows us to demonstrate formally that if two observers moving STL with respect to one another see time moving in opposite directions, we can trivially generate a time travel situation  (see figure \ref{fig-tt}, right hand side).

Suppose, then, that $o_1$ and $o_2$ are observers moving STL relative to $o$, and that the future quadrant of $o_1$ is \TOP (say). If $o_2$ moves STL relative to $o_1$, its own future quadrant will necessarily be either \TOP or \BOTTOM, depending on the direction in which it considers time to flow. Of course, if $o_2$'s future quadrant is \BOTTOM (\ie time for $o_2$ `flows backwards' as far as $o_1$ is concerned), we can trivially construct a time travel scenario by \AxSTL. Consequently, since we are interested in whether FTL motion entails the possibility of time travel in an otherwise well-behaved setting, we will assume that the future quadrant for $o_2$ is also \TOP. Analogous arguments apply for other choices of $o_1$'s future quadrant, so we will make the blanket assumption:
\begin{itemize}
\item 
  \AxSameFuture\\
  Whenever two observers travel slower than light relative to one another, they agree as to the direction of time's flow (they have the same future quadrant).
\end{itemize}

A simple argument now shows that FTL observers (with respect to $o$) also agree with one another as to the direction of time's arrow. For suppose that observers $o_1$ and $o_2$ are both traveling FTL relative to $o$. As before, we can assume that $o$'s future quadrant is \TOP, and since $o_1$ and $o_2$ are moving FTL relative to $o$, their own future quadrants must be either \LEFT or \RIGHT. They are therefore moving slower than light relative to one another, and \AxSameFuture applies.

Since the direction in which time flows is fixed along any given worldline, any manifestation of time travel or its associated causality violations must involve interactions between two or more observers. So we say that \emph{time  travel is possible} (Figure \ref{fig-tt}) to mean there is a sequence $o_1, o_2, \dots, o_n$ of (at least 2 distinct) observers, and a sequence $e_0, e_1, \dots, e_n$ of events, such that for each $1\le i\le n$,
\begin{itemize}
\item
   $o_1 \neq o_n$;
\item
   $e_0 = e_n$;
\item
   $e_{i-1}$ and $e_i$ are both events on $o_i$'s worldline;
\item
   $e_{i-1}$ chronologically precedes $e_i$ according to $o_i$.
\end{itemize}
We say that time travel can be \textit{implemented using FTL observers} if not all of the observers occurring in $o_1,\dots,o_n$ move STL relative to one another.
\begin{figure}
  \centering
  \includegraphics*{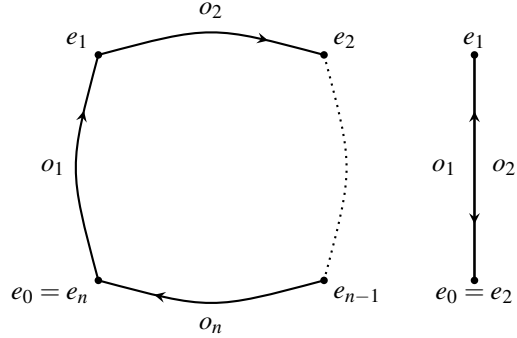}
  \caption{A general time travel scenario (left), and the particular case involving just two observers (right).}
  \label{fig-tt}
\end{figure}

We can now state and prove our first theorem. The fact that its conditions do not make Theorem~\ref{thm:main} vacuously true is explained in Remark 1 below.

\goodbreak
\begin{theorem}\label{thm:main}
Assume \AxCMV. Then in \SPACETIME,
\begin{itemize}
\item[(a)] 
  \label{item-nott} 
  \AxSameFuture implies that time travel is impossible; and hence,
\item[(b)] 
  \label{item-nocv} 
  if \AxSTL holds, time travel can be implemented using FTL observers only if it is already possible without them.
\end{itemize}
\end{theorem}

\begin{proof}
(a) By \AxSameFuture, we may assume, without loss of generality, that the future quadrant of STL observers is \TOP, and that of FTL observers is \RIGHT. (Here,  STL and FTL are understood relative to our fixed observer $o$.) Suppose observers $o_1,\ldots, o_n$ and events $e_0, \dots, e_n$ can be chosen which implement a time travel situation (we will show that this assumption contradicts \AxSameFuture). Let $o_i$ be any of the observers ($1\le i\le n$), and consider the lightcone at the event $e_{i-1}$ on $o_i$'s worldline. This cone comprises two lines, one going `up to the right' (from \BOTTOM-\LEFT to \TOP-\RIGHT), the other going up to the left. We will call the latter the ``main axis'' at $e_{i-1}$ (see figure \ref{fig-proof} for the main axis at $e_0$). We will write $d_i$ to mean the perpendicular (Euclidean) distance from $e_i$ to this main axis, assigning it a positive sign if $e_i$ is above the main axis and a negative sign if it is below. Notice that all main axes are parallel, and that distance is additive (if $a$, $b$ and $c$ are events, the distance from $c$ to $a$'s main axis can be calculated by summing $b$'s distance from $a$'s axis with $c$'s distance from $b$'s axis).
\begin{figure}[!ht]
  \centering
  \includegraphics*{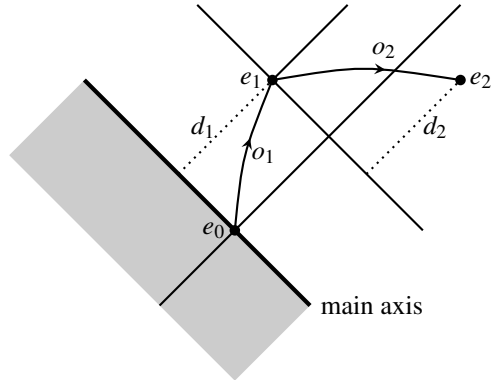}
  \caption{An observer's motion always acts to increase the perpendicular Euclidean distance from the main axis, regardless of whether it is FTL or STL (illustration for the proof of Theorem \ref{thm:main}).}
  \label{fig-proof}
\end{figure}

By assumption, the future quadrant of $o_i$ at $e_{i-1}$ is either \TOP or \RIGHT, from which it follows that $e_i$ must lie in the interior of one of these two quadrants. Both of these quadrants are bounded below/to the left by the main axis, so in each case we will have $d_i > 0$. It follows from the additivity of distances that the perpendicular distance from $e_n$ to $e_0$'s axis is equal to $d_1 + \dots + d_n$, which (being a sum of positive terms) is again positive. But this contradicts the requirement that $e_0 = e_n$ in a time travel scenario. Hence time travel is impossible if \AxSameFuture and \AxCMV are assumed at the same time, as claimed.

(b) If time travel is possible, then (a) tells us that \AxSameFuture must be violated, so there must be two observers traveling slower than light relative to one another who disagree on the direction of time's flow. However, if time flows in different directions for two observers, then time travel is already trivially possible by \AxSTL.
\end{proof}
\bigskip

In the following discussions, by a \emph{model}, we understand \SPACETIME\ together with some distinguished set of worldlines. By an automorphism $P$ of this model, we understand a permutation of \SPACETIME that preserves the absolute value of Minkowski metric as well as distinguished worldlines (\ie $w$ is distinguished if and only if $P\circ w$ is distinguished). We say that Einstein's \emph{Special Principle of Relativity} (\SPR) is satisfied in a model if any inertial observer in this model can be taken to any other by an automorphism (recall that an inertial observer is one with a straight non-light-like worldline).
\bigskip

\noindent {\bf Remark 1.} We can construct a model in which both \AxSameFuture and \AxCMV are true, in which Einstein's \emph{Special Principle of Relativity} (\SPR) is satisfied, and in which all kinds of STL as well as FTL relative motion are possible. In this model, there is no time travel by Theorem~\ref{thm:main}.

Briefly, this model is as follows. Take the (1+1)-dimensional spacetime for special relativity with all kinds of directed worldlines, timelike as well as spacelike, inertial as well as non-inertial. Of these, keep only those worldlines that respect the (absolute value of the) Minkowski metric of \SPACETIME. This way \AxCMV is satisfied in the resulting model. Take any coordinate system (for \SPACETIME), and of these observers, keep only those timelike ones whose future quadrant is \TOP, of the spacelike observers keep only those whose future quadrant is \RIGHT. This way \AxSameFuture is satisfied, too.
This model satisfies \SPR because of the following. Any two spacelike or timelike inertial observers can be taken to each other by a Poincar\'e transformation $P$, and any timelike observer can be taken to a spacelike one by a Poincar\'e transformation composed with the transformation $T$ which just interchanges the first coordinate with the second one in our fixed coordinate system for \SPACETIME. It is easy to check that both $P$ and $T$ are automorphisms of our model.

Let us fix an arbitrary inertial observer; what does the world look like in its coordinate system? (By \SPR, the world looks exactly the same in all other inertial worldviews.) As far as STL observers are concerned, the world is completely normal, it is as we know it from special relativity: moving clocks run slow, etc. But there is also a time-dilation effect for FTL observers: the closer their speed is to the speed of light the slower their clocks run, but the greater their speed is, the faster their clocks run. In addition to this, in one of the two spatial directions (call it $-x$) FTL clocks run backwards.
How can we interpret this phenomenon? If we consider the time orientation of an observer to be the direction in which his own future unfolds, we can say that information can be communicated between events on his worldline only in the direction of this unfolding future (for a discussion of this kind of intuition, see, \eg \cite{New70}).
Expressing this in spacetime terminology, the ``causal past'' of an event $e_0$ in our model is the region ``below'' its main axis (see the shaded area in Figure \ref{fig-proof}). The unusual thing here is that this causal past does not lie entirely within the observer's temporal past (\ie the half-space below the simultaneity).

Keeping this description in mind, let's see how we can resolve the paradox described by Newton \cite{New70} (see Figure~\ref{ccc}). In the model we constructed above, the ``reading'' of Figure~\ref{ccc} is as follows. At event $B$ observer $o_1$ sends out FTL signal $s_1$. This is received by observer $o_2$ at event $C$, who later at event $D$ receives another signal, $s_2$ from $o_1$. (The only unusual thing is that $o_2$ receives the two signals in reverse time-order
from $o_1$'s point of view.) If $o_2$ wishes to send $o_1$ a reply to $s_1$ at event $D$, this reply has to be an STL
signal. \QED

\subsection{The general ($n$+1)-dimensional case}
\label{sec:n+1}

We now turn to generalizing our model to the ($n$+1)-dimensional case, for all $n \ge 1$. We will have to make some adjustments because ($n$+1)-dimensional Minkowski spacetime $\nSPACETIME$ is very different from \SPACETIME when $n > 1$. Here, unlike in $\SPACETIME$, it makes sense to distinguish between STL and FTL motion in absolute terms, because the interior of a lightcone is geometrically (and hence physically) distinguishable from its outside. Because STL inertial motion is geometrically distinguishable from FTL inertial motion, we cannot attach reference frames to FTL objects in a way that respects \SPR. (A strictly axiomatic derivation of this fact can be found, \eg in \cite[Thm.~2.1]{AMNSz12}).
Nonetheless, we will define spacelike smooth functions $w \colon \REAL \to \nSPACETIME$ parametrized by their arc lengths to be worldlines of FTL particles (or signals). Thus, particles can ``move FTL'', they have ``clocks'', but they do not have reference frames attached to them. Timelike smooth functions remain worldlines of observers. With this terminology, our previous definitions of time travel, model, etc., remain meaningful in $\nSPACETIME$ for $n > 1$.
In particular, by a \emph{model}, we understand $\nSPACETIME$ together with some distinguished set of worldlines. By an automorphism of a model, we understand a permutation of $\nSPACETIME$ that respects both the Minkowski metric and distinguished worldlines. By an inertial observer, we mean an observer whose worldline is a timelike straight line. We say that a model satisfies Einstein's \SPR if any inertial observer can be taken to any other by an automorphism. By the \emph{standard model of special relativity}, we understand \nSPACETIME together with all smooth future-directed timelike curves $w \colon \REAL \to \nSPACETIME$ parametrized by their arc lengths as worldlines of observers.

\begin{theorem}\label{thm:mainplus}
There is an enrichment of the ($n$+1)-dimensional standard model of special relativity with FTL particles in which
\begin{itemize}
\item 
  Einstein's \SPR is satisfied;
\item 
  each spacelike line is the worldline of an FTL particle;
\end{itemize}
\noindent but time travel is impossible.
\end{theorem}

\begin{proof}
For simplicity, we only prove the $n=3$ case because it is straightforward to generalize this proof for arbitrary $n$. Informally, the construction is a generalization of the (1+1)-dimensional one outlined in Remark 1, in which we replace the light-like line $\{ (t,-t) : t\in\REAL\}$ (the main axis) with a Robb hyperplane $R$ (namely, the hyperplane Minkowski-orthogonal to this light-like line), and we orient the lines within the Robb hyperplane so that there cannot be any time-travel within it. Then we define the orientation on a line $\ell$ such that ``it goes from below $R$ to above it'' if it is not parallel to $R$, and otherwise we copy the orientation of $R$ to the line $\ell$.
\begin{figure}
  \centering
  \includegraphics*{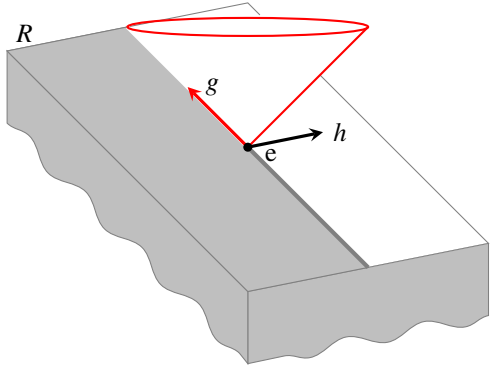}
  \caption{The construction in the proof of Theorem~\ref{thm:mainplus} for (2+1)-dimensions.}
  \label{FigureB}
\end{figure}

We write down the construction formally, writing $\mu$ for the ((3+1)-dimensional) Minkowskian scalar product. For fixed vector $u$, the set $\{ v : \mu(u,v)=0\}$ is a hyperplane (just as in the Euclidean case), and $\{ v : \mu(u,v)>0\}$ is an open half-space -- we can think of it as the half of the space ``above'' the hyperplane.

To formalize our construction, let $g$, $h$, $s$ be three mutually orthogonal vectors (\ie $\mu(g,h)=\mu(g,s)=\mu(h,s)=0$) such that $g$ is light-like, and $h$, $s$ are spacelike.  The Robb hyperplane in our construction will be determined by $g$, and we will use $h$, $s$ for giving a `good' orientation in this Robb hyperplane. We define the set $C$ of \emph{positive vectors} as follows: we define $v$ to be positive ($v\in C$) if and only if
\begin{description}
\item[] $\mu(g,v)>0$, \quad or
\item[] $\mu(g,v)=0$ and $\mu(h,v)>0$, \quad or
\item[] $\mu(g,v)=\mu(h,v)=0$ and $\mu(s,v)>0$, \quad or
\item[] $v=\lambda g$ for some $\lambda>0$.
\end{description}
Intuitively, a vector is positive if and only if it (points from the origin to a point that) is above the Robb hyperplane determined by $g$, or it is in the open half of this Robb hyperplane which is above the hyperplane determined by $h$, or it is in the half of the intersection of the Robb hyperplane and the one determined by $h$ which is above the hyperplane determined by $s$, or it lies on the half-line determined by $g$. See Figure~\ref{FigureB}.

We include in our model as worldlines exactly those smooth functions $w \colon \REAL \to \tSPACETIME$ whose tangent vectors are positive in the above-defined sense, and which are parametrized by their arc lengths. For example, each spacelike straight line is a worldline, as it is not difficult to see that $C$ together with its central inversion $-C$ cover the whole space except the zero-vector.

To check that \SPR holds, we have to find an automorphism $P$ for any two timelike lines $\ell$ and $\ell'$ which takes $\ell$ to $\ell'$. Take any point $p$  on $\ell$. Let $x$, $y$ and $z$ be the lines going through $p$, orthogonal to $\ell$, and lying in the planes  $\ell g$, $\ell h$ and $\ell s$ respectively (where we write $\ell v$ for the plane that is parallel to vector $v$ and contains $\ell$ if there is only one such plane, \ie if $\ell$ and $v$ are not parallel); define $x'$, $y'$ and $z'$ analogously to be lines orthogonal to $\ell'$ in the planes $\ell' g$, $\ell' h$ and $\ell' s$. Then the lines $\ell$, $x$, $y$ and $z$ are mutually orthogonal, as are the lines $\ell'$, $x'$, $y'$ and $z'$. Now, let $P$ be the Poincar\'e transformation that takes $\ell$, $x$, $y$, $z$ to  $\ell'$, $x'$, $y'$, $z'$ and respects the orientation determined by $C$ (such a $P$ clearly exists). This $P$ respects the  Minkowski metric since it is a Poincar\'e transformation. It also respects time-orientation because it takes any plane parallel to the plane $gh$ to another such plane (because $gh$, $xy$ and $x'y'$ are all the same plane), and similarly for planes parallel to $gs$.

Now, time travel is not possible in this model for essentially the same reason it is not possible in our (1+1)-dimensional model. For suppose we have a time-travel scenario linking the events $e_0, \dots, e_n$. We will show that if $w$ is any observer/particle, then $w(b) - w(a) \in C$ whenever $a < b$. It will then follow that each $e_{i+1} - e_i$ is in $C$, and hence (because $C$ is closed under vector addition) that $e_n - e_0 = (e_n - e_{n-1}) + (e_{n-1} - e_{n-2}) + \dots + (e_1 - e_0)$ is in $C$. Since $C$ does not contain the zero vector, it follows that $e_n \neq e_0$, contradicting the definition of a time-travel scenario, and the result follows.  

It remains to show that $w(b) - w(a) \in C$ whenever $a < b$ and $w$ is one of the worldlines included in our model. We can assume without loss of generality that $a = 0$, and we know that the ``component of $w$ orthogonal to $R$'' is non-decreasing (its derivative is non-negative by positivity of tangent vectors). Consequently, it is enough to prove that for all $v > 0$: if $w(u) \in \nSPACETIME$ whenever $0 < u < v$, then $w(u) \in C$ whenever $0 < u < v$. We will argue by induction on $n$. We have already shown the result to be true in the (1+1)-dimensional case, \ie for $n = 1$, so let us assume that the result is valid in ($k$+1)-dimensional spacetime for some $k \geq 1$, and consider the case $n = k+1$. Choose any $b > 0$. If the component of $w$ orthogonal to $R$ is non-constant on $[0,b]$, then $w(b)$ is obviously in $C$. On the other hand, if this component is constant on $[0,b]$, then $w(t)$ lies entirely in the ($k$+1)-dimensional space $R$ for all $t \in [0,b]$, and the result follows by induction. 
\end{proof}
\bigskip

\noindent {\bf Remark 2.} We describe how the world looks in the coordinate system $F_o$ of an arbitrary STL inertial observer $o$ in the (3+1)-dimensional model constructed in the proof of Theorem~\ref{thm:mainplus}. From the point of view of STL motion everything is normal in this world view, \ie everything is as special relativity prescribes. The novelty is that there are particles moving FTL, and they have ``clocks'' inducing a time orientation on their worldlines. Following Newton~\cite{New70}, let us call them signals (objects able to carry information in the direction of their time orientation). Some of these signals are such that their time orientation is opposite to the one of our chosen coordinate system $F_o$. Seen from this reference frame, information flows backwards in these signals. Or, in other words, the clock of this FTL signal runs backward as seen by $F_o$. Newton~\cite{New70} explains how this may be possible, and he even gives clues for how one could detect such a backward information-flow experimentally. Let us call such signals \emph{inverse} signals.

We now describe how FTL signals behave in terms of \emph{space} and \emph{motion}, in our coordinate system $F_o$. We will see that the inverse signals distinguish three pairwise orthogonal directions $x$, $y$, $z$ in space. First, there is a unique spatial direction, call it $-x$, in which there is a maximal effect of backward-flow: in direction $-x$ each FTL signal is an inverse one (this means that all FTL signals on a line parallel to $x$ go in direction $x$). On the other hand, in spatial directions orthogonal to $x$ there is a minimal effect of backward-flow: only an infinitely rapid signal can be an inverse one. What happens in between? In any other direction $u$, there is a threshold FTL speed $v$ such that backward-flow can occur only in one of the directions $u$ and $-u$, and in this direction exactly the signals faster than $v$ are inverse ones. In fact, there is a simple formula for the threshold speed in direction $u$: we may assume that $|u|=1$ and also $|x|=1$, and now the threshold speed in direction $u$ is $-1\slash (u \cdot x)$, where $u \cdot x$ denotes the Euclidean scalar product of $u$ and $x$ and we have taken the speed of light to be 1. Figure~\ref{FigureA} provides an illustration for the threshold velocity in direction $u$.
\begin{figure}[!ht]
  \centering
  \includegraphics*{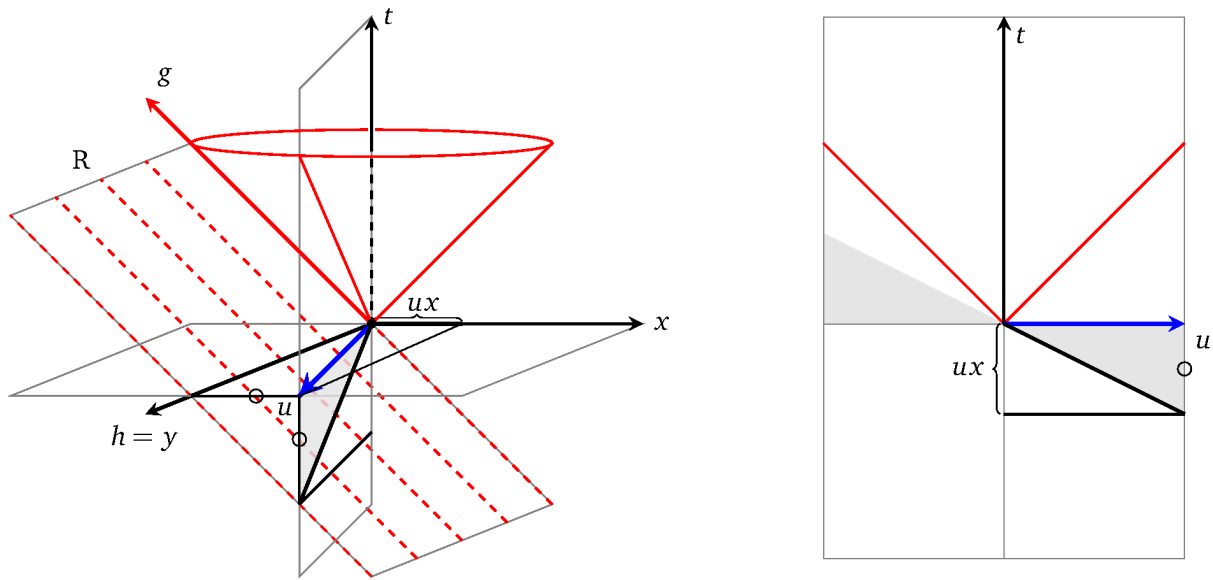}
  \caption{The threshold velocity in direction $u$.} 
  \label{FigureA}
\end{figure}

To complete the description of FTL signals in $F_o$, it only remains to describe the time orientations of infinite-speed clocks in directions orthogonal to $x$.  Let $S$ denote the spatial plane orthogonal to $x$. Lines in $S$ are directed as in our construction: there are two orthogonal vectors in $S$, call them $y$ and $z$, and lines in $S$ are directed to go from one half-plane determined by $y$ to the other, while lines parallel to $z$ are directed to go in direction $z$.

We can summarize this picture as saying that Einstein's simultaneity is an appropriate one for STL motion, while for FTL motion the Robb hyperplane of our construction seems more appropriate as a simultaneity. \QED

\section{Concluding remarks}
\label{sec:concluding-remarks}

In this paper, we have exhibited a model of \nSPACETIME which is
inhabited by observers/particles moving at all (non-light)
speeds relative to one another, but in which time travel is not
possible -- it follows that the existence of FTL
signals \emph{does not} logically entail the
existence of `time travel' scenarios. Nor, therefore, does it
inevitably lead to the causality paradoxes arising from those
scenarios. Our model is, moreover, physically sensible since it
satisfies Einstein's Special Principle of Relativity -- each 
inertial observer can be taken to any other by
an automorphism.

On the other hand, it is certainly possible to construct logically
sensible models of spacetime in which observers can disagree as to
the direction of time's arrow. Our results do not undermine those
constructions, but they do force us to re-examine which aspects of
these models are actually responsible for any apparent paradoxes.

\section*{Acknowledgement}
This work is supported under the Royal Society International
Exchanges Scheme (reference IE110369) and by the Hungarian
Scientific Research Fund for basic research grants No.~T81188 and
No.~PD84093, as well as by a Bolyai grant for Madar\'asz. This work
was partially completed while Stannett was a visiting fellow at the
Isaac Newton Institute for Mathematical Sciences, under the program
\textit{Semantics \& Syntax: A Legacy of Alan Turing}. 
The authors are grateful to the editor and anonymous referees for their care, comments and suggestions.

\bibliographystyle{spmpsci}   
\bibliography{Revised-FTLCQG}

\begin{thebibliography}{10}
\providecommand{\url}[1]{{#1}}
\providecommand{\urlprefix}{URL }
\expandafter\ifx\csname urlstyle\endcsname\relax
  \providecommand{\doi}[1]{DOI~\discretionary{}{}{}#1}\else
  \providecommand{\doi}{DOI~\discretionary{}{}{}\begingroup
  \urlstyle{rm}\Url}\fi

\bibitem{Akl10}
Akl, S.G.: {Time Travel: A New Hypercomputational Paradigm}.
\newblock Int. J. Unconventional Computing \textbf{6}(5), 329--351 (2010)

\bibitem{AMN02}
Andr{\'e}ka, H., Madar{\'a}sz, J.X., N{\'e}meti, I.: On the logical structure
  of relativity theories.
\newblock \url{www.math-inst.hu/pub/algebraic-logic/Contents.html}. 1312pp.
  (July 5, 2002).
\newblock With contributions from A. Andai, G. S{\'a}gi, I. Sain and Cs.
  T\H{o}ke

\bibitem{AMNSz12}
Andr{\'e}ka, H., Madar{\'a}sz, J.X., N{\'e}meti, I., Sz{\'e}kely, G.: {A logic
  road from special relativity to general relativity}.
\newblock Synthese \textbf{186}(3), 633--649 (2012)

\bibitem{ANS12}
Andr{\'e}ka, H., N{\'e}meti, I., Sz{\'e}kely, G.: {Closed Timelike Curves in
  Relativistic Computation}.
\newblock Parallel Processing Letters \textbf{22}, 1240,010 (2012).
\newblock \doi{10.1142/S0129626412400105}.
\newblock (\textit{Preprint} gr-qc/1105.0047)

\bibitem{Bru03}
Brun, T.A.: Computers with closed timelike curves can solve hard problems
  efficiently.
\newblock Foundations of Physics Letters \textbf{16}(3), 245--253 (2003)

\bibitem{BW12}
Brun, T.A., Wilde, M.M.: {Perfect state distinguishability and computational
  speedups with postselected closed timelike curves}.
\newblock Foundations of Physics \textbf{42}(3), 341--361 (March 2012)

\bibitem{PC10}
Calude, C.S., Costa, J.F., Guerra, H.: {Towards a Computational Interpretation
  of Physical Theories}.
\newblock Applied Mathematics and Computation \textbf{219}, 1--442 (15
  September 2012)

\bibitem{Deu91}
Deutsch, D.: {Quantum mechanics near closed timelike lines}.
\newblock Phys. Rev. D \textbf{44}(10), 3197–--3217 (15 November 1991)

\bibitem{Ein07}
Einstein, A.: {{\"U}ber die vom Relativit\"atsprinzip geforderte Tr\"agheit der
  Energie}.
\newblock Annalen der Physik \textbf{328}(7), 371--–384 (1907)

\bibitem{EN02}
Etesi, G., N{\'e}meti, I.: {Non-Turing computations via Malament-Hogarth
  space-times}.
\newblock International Journal of Theoretical Physics \textbf{41}, 341--370
  (2002).
\newblock (\textit{Preprint} gr-qc/0104023)

\bibitem{Nov90}
Friedman, J., Morris, M.S., Novikov, I.D., Echeverria, F., Klinkhammer, G.,
  Thorne, K.S., Yurtsever, U.: {Cauchy problem in spacetimes with closed
  timelike curves}.
\newblock Phys. Rev. D \textbf{42}(6), 1915--1930 (1990)

\bibitem{Hog04}
Hogarth, M.: {Deciding Arithmetic using SAD Computers}.
\newblock The British Journal for the Philosophy of Science \textbf{55},
  681--691 (2004)

\bibitem{Llo11}
Lloyd, S., et~al.: {Closed Timelike Curves via Postselection: Theory and
  Experimental Test of Consistency}.
\newblock Phys. Rev. Letters \textbf{106}, 040,403 (2011)

\bibitem{Nov92}
Lossevtt, A., Novikov, I.D.: {The Jinn of the time machine: non-trivial
  self-consistent solutions}.
\newblock Class. Quantum Grav. \textbf{9}, 2309--2321 (1992)

\bibitem{MNS06}
Madar{\'a}sz, J.X., N{\'e}meti, I., Sz{\'e}kely, G.: {Twin Paradox and the
  Logical Foundation of Relativity Theory}.
\newblock Foundation of Physics \textbf{36}(5), 681--714 (2006)

\bibitem{New70}
Newton, R.G.: {Particles That Travel Faster than Light?}
\newblock Science \textbf{167}(3925), 1569--1574 (1970)

\bibitem{Recami87}
Recami, E.: {Tachyon kinematics and causality: a systematic thorough analysis
  of the tachyon causal paradoxes}.
\newblock Foundations of Physics \textbf{17}(3), 239--296 (1987)

\bibitem{Sta06}
Stannett, M.: {The case for hypercomputation}.
\newblock Applied Mathematics and Computation \textbf{178}(1), 8--24 (2006)

\bibitem{Sta12}
Stannett, M.: {Computing the appearance of physical reality}.
\newblock Applied Mathematics and Computation \textbf{219}(1), 54--62 (2012)

\bibitem{Sta13}
Stannett, M.: {Computation and Spacetime Structure}.
\newblock Int. J. Unconventional Computing \textbf{9}(1--2), 173--184 (2013)

\bibitem{SN12}
Stannett, M., N\'emeti, I.: {Using Isabelle/HOL to Verify First-Order
  Relativity Theory}.
\newblock Journal of Automated Reasoning pp. 1--18 (2013).
\newblock \doi{10.1007/s10817-013-9292-7}.
\newblock \urlprefix\url{http://dx.doi.org/10.1007/s10817-013-9292-7}.
\newblock (\textit{Preprint} cs-LO/1211.6468)

\bibitem{Tol17}
Tolman, R.C.: {The Theory of the Relativity of Motion}.
\newblock University of California Press (1917)

\end{thebibliography}

\end{document}